\documentclass[11pt]{article}

\usepackage{times}
\usepackage{latexsym}
\usepackage{amsfonts,amsthm,amssymb}
\usepackage{amsmath}
\usepackage{euscript}
\usepackage{amstext}
\usepackage{graphicx}
\usepackage{color}
\usepackage{url,hyperref}
\usepackage{dsfont}

\newtheorem{lemma}{Lemma}[section]
\newtheorem{theorem}[lemma]{Theorem}

        {\hspace*{\fill}$\Box$\par}

\newcommand{\opt}{\textrm{\sc OPT}}
\newcommand{\etal}{et al.\ }
\newcommand{\eps}{\epsilon}
\newcommand{\Algorithm}[1]{\textrm{\sc #1}} 

\def\Set#1{\left\{ #1 \right\}}
\def\Paren#1{\left( #1 \right)}	
\def\Brack#1{\left[ #1 \right]}

\newcommand{\qs}{\mathcal{Q^S}(t)}
\newcommand{\ws}{\mathcal{W^S}(t)}
\newcommand{\qo}{\mathcal{Q^O}(t)}
\newcommand{\cs}[1]{C_{#1}^S}
\newcommand{\co}[1]{C_{#1}^O}
\newcommand{\rs}[1]{R^S(#1,t) }
\newcommand{\vo}[1]{V^O(#1,t) }
\newcommand{\rsc}[1]{R^S(#1,\cs{#1}) }
\newcommand{\voc}[1]{V^O(#1,\cs{#1}) }

\newcommand{\ps}[1]{p^S_{#1}(t) }
\newcommand{\po}[1]{p^O_{#1}(t) }
\newcommand{\srpt}{\Algorithm{SRPT}}
\newcommand{\dt}{\mathrm{dt}}

\newcommand{\dds}{\frac{\mathrm{d}}{\dt} \srpt}
\newcommand{\ddphi}{\frac{\mathrm{d}}{\dt} \Phi}

\newcommand{\sjobs}[1]{\mathcal{S}_{#1}}
\newcommand{\tjobs}[1]{\mathcal{T}_{#1}}

\begin{document}
\title{Online Scheduling on Identical Machines using SRPT\thanks{This paper contains results that appeared in the preliminary version \cite{FoxM11}.}}

\author{
Kyle Fox\thanks{Department of Computer Science, University of
Illinois, 201 N.\ Goodwin Ave., Urbana, IL 61801. {\tt
kylefox2@illinois.edu}. This work was done while the author was at Google Inc.
(Mountain
View).} \and Benjamin
Moseley\thanks{Department of Computer Science, University of
Illinois, 201 N.\ Goodwin Ave., Urbana, IL 61801. {\tt
bmosele2@illinois.edu}. This work was done while the author was at Yahoo! Labs
(Santa
Clara).} }
\date{}
\maketitle

\begin{abstract}

Due to its optimality on a single machine for the problem of minimizing average
flow time, Shortest-Remaining-Processing-Time (\srpt) appears to be the most
natural algorithm to consider for the problem of minimizing average flow time on
multiple identical machines. It is known that $\srpt$ achieves the best possible
competitive ratio on multiple machines up to a constant factor. Using resource
augmentation, $\srpt$ is known to achieve total flow time at most that of the
optimal solution when given machines of speed $2- \frac{1}{m}$. Further, it is
known that $\srpt$'s competitive ratio improves as the speed increases; $\srpt$
is $s$-speed $\frac{1}{s}$-competitive when $s \geq 2- \frac{1}{m}$.

However, a gap has persisted in our understanding of $\srpt$. Before this work,
the performance of $\srpt$ was not known when $\srpt$ is given $(1+\eps)$-speed
when $0 < \eps < 1-\frac{1}{m}$, even though it has been thought that $\srpt$ is
$(1+\eps)$-speed $O(1)$-competitive for over a decade. Resolving this question
was suggested in Open Problem 2.9 from the survey ``Online Scheduling'' by
Pruhs, Sgall, and Torng \cite{PruhsST}, and we answer the question in this paper.
We show that $\srpt$ is \emph{scalable} on $m$ identical machines. That is, we
show $\srpt$ is $(1+\eps)$-speed $O(\frac{1}{\eps})$-competitive for $\eps
>0$. We complement this by showing that $\srpt$ is $(1+\eps)$-speed
$O(\frac{1}{\eps^2})$-competitive for the objective of minimizing the
$\ell_k$-norms of flow time on $m$ identical machines. Both of our results rely
on new potential functions that capture the structure of \srpt. Our results,
combined with previous work, show that $\srpt$ is the best possible online
algorithm in essentially every aspect when migration is permissible.

\end{abstract}

\setcounter{page}{0} \thispagestyle{empty} \clearpage

\section{Introduction}

Scheduling jobs that arrive over time is a fundamental problem faced by a
variety of
systems. In the simplest setting there is a single machine and $n$ jobs that
arrive online. A job $J_i$ is released at time $r_i$ where $i \in [n]$. The job
has some processing
time $p_i$. This is the amount of time the scheduler must devote to job $J_i$ to
complete
the job. The goal of the scheduler is to determine which job should be processed
at any given time while
optimizing a quality of service metric. In the \emph{online} setting the
scheduler is not
aware of a job until it is released. Thus, an online scheduler must make
scheduling
decisions without access to the entire problem instance. Having the scheduler be
online is
desirable in practice since most systems are not aware of the entire jobs
sequence in advance.

The most popular quality of service metric considered in online scheduling
theory is total
flow time, or equivalently, average flow time~\cite{PruhsST}.
The flow time\footnote{Flow time
is also
referred to as response time or waiting time.} of a job is the amount of time it
takes the
scheduler to satisfy the job. Formally, the flow time of job $J_i$ is $C_i -
r_i$ where $C_i$ is
the completion time of job $J_i$. The completion time of a job $J_i$ is defined
to be the
earliest time $t$ such that the scheduler has devoted $p_i$ units of time to job
$J_i$
during $(r_i, t]$. The total flow time of the schedule is $\sum_{i \in [n]} C_i
- r_i$. By
focusing on minimizing the total flow time, the scheduler minimizes the total
time jobs must
wait to be satisfied.

On a single machine, the algorithm Shortest-Remaining-Processing-Time ($\srpt$)
always
schedules the job whose remaining processing time is the smallest, breaking ties
arbitrarily. It is well known that
$\srpt$ is optimal for total flow time in this setting.  A more 
complicated scheduling model is where there are $m$ identical machines.
Minimizing the flow time in this model has
been studied extensively in scheduling theory 
\cite{LeonardiR07,AwerbuchALR02,ChekuriGKK04,AvrahamiA07,BecchettiL04,ChekuriKZ01,TorngM08}.
When there is more than one machine the
scheduler must not only chose which subset of jobs to schedule,
but it must also decide how to distribute jobs across the machines.
Naturally, it is assumed that
a job can only be processed by one machine at a time. For this scheduling
setting, it is known that there there is
a $\Omega(\min\{\log P , \log n/m \})$ lower bound on any online randomized
algorithm in
the oblivious adversary model \cite{LeonardiR07}. Here $P$ is the ratio of
maximum
processing time to minimum processing time. The algorithm $\srpt$ in the $m$
identical machine setting always schedules the $m$ jobs with least remaining
processing time. $\srpt$ has competitive ratio $O(\min\{\log P , \log n/m \})$
for average flow time, making $\srpt$ the best possible algorithm up to a
constant factor in the competitive ratio.

The strong lower bound on online algorithms 
has led previous work to use a resource augmentation analysis. In a
\emph{resource augmentation} analysis the adversary is given $m$ unit-speed
processors and the algorithm is given $m$ processors of speed $s$
\cite{KalyanasundaramP95}. We say that an algorithm is $s$-speed $c$-competitive
if the algorithm's objective is within a factor of $c$ of the optimal solution's
objective when the algorithm is given $s$ resource augmentation. An ideal
resource augmentation analysis shows that an algorithm is $(1+ \eps)$-speed
$O(1)$-competitive for any fixed $\eps > 0$. Such an algorithm is called
\emph{scalable}. A scalable algorithm is $O(1)$-competitive when given the
minimum amount of
extra resources over the adversary. Given the strong lower bound on flow time in
the
identical machines model, finding a scalable algorithm is essentially the best
positive result
that can be shown using worst case analysis.

Given that $\srpt$ is an optimal algorithm on a single machine and achieves the
best
possible competitive ratio on multiple machines without resource augmentation,
it was
widely thought that $\srpt$ would be a scalable algorithm in the multiple
machine case.
However, the competitive ratio of $\srpt$ when given $1+\eps$ speed had been
unresolved for about a decade when $0 < \eps < 1-\frac{1}{m}$. Instead, another
algorithm was shown to be scalable \cite{ChekuriGKK04}. This algorithm
geometrically groups jobs according to their size. It uses these groupings to
assign each job to exactly one machine. The algorithm then runs the single
machine version of $\srpt$ separately on each machine.

Although the competitiveness of $\srpt$ was not 
known when given speed less than $2-\frac{1}{m}$,  it was known that $
\srpt$ achieves total flow time at most that of the optimal solution's flow time
when given
machines of speed at least $2 - \frac{1}{m}$ \cite{PhillipsSTW02}. In fact, this
has been extended to show that $ \srpt$ is $s$-speed $\frac{1}{s}$-competitive
when $s \geq 2- \frac{1}{m}$ \cite{TorngM08}.
This result shows that $\srpt$ `efficiently' uses the faster processors it is
given. In the fairly
recent online scheduling survey of Pruhs, Sgall, and Torng it was suggested in
Open
Problem 2.9 that an important question is to resolve whether or not $\srpt$ is a
scalable
algorithm \cite{PruhsST}. In this paper we answer this question in the
affirmative by
showing the following theorem.  

\begin{theorem}\label{thm:avg}
The algorithm $\srpt$ is $(1+\eps)$-speed
$\frac{4}{\eps}$-competitive for average flow
time on $m$ identical parallel machines for $\eps >0$. 
\end{theorem}

Unfortunately, algorithms which are optimal for average flow time can starve
individual
jobs of processing power for an arbitrary finite amount of time. For example,
suppose we are given a single machine. Jobs $J_1$ and $J_2$ arrive time $0$ and
at every unit time step another job arrives. All jobs have unit processing time.
Using average flow time as the objective, an optimal algorithm for this problem
instance is to schedule $J_1$ and then schedule jobs as they arrive, scheduling
$J_2$ after the last of the other jobs is completed. Although this algorithm is
optimal, it can be seen that the algorithm is not `fair' to job $J_2$.

Algorithms which are fair at the individual job level are desirable in practice
\cite{Tanenbaum07,SilberschatsG94}. In fact, algorithms that are competitive for
total flow time are sometimes not implemented due to the possibility of
unfairness \cite{BansalP03}. To overcome the disadvantage of algorithms that
merely optimize the average flow time, the objective of minimizing the
$\ell_k$-norms of flow time for small $k$ was suggested by Bansal and Pruhs
\cite{BansalP03,BansalP04}. This objective tries to balance overall performance
and fairness. Specifically, the $\ell_k$-norm objective minimizes $\left
(\sum_{i \in [n]} (C_i - r_i)^k \right )^{1/k}$. Notice that optimizing the
$\ell_1$-norm is equivalent to optimizing the average flow time. For the
$\ell_k$-norm objective when $k >1$, the previous example has one optimal
solution. This solution schedules jobs in the order they arrive, which can be
seen to be `fair' to each job.

For the $\ell_k$-norm objective it is known that every online deterministic
algorithm is $n^{\Omega(1)}$-competitive even on a single machine when $1 < k <
\infty$ \cite{BansalP03}. This is quite different from the $\ell_1$-norm where
$\srpt$ is an optimal algorithm. In the single machine setting, it was shown
that $\srpt$ is a scalable algorithm for the $\ell_k$-norm objective for all $k$
\cite{BansalP03}. The competitiveness of $\srpt$ in the multiple machine setting
was not known for the $\ell_k$-norms even when $\srpt$ is given any constant
amount of resource augmentation. The previously discussed algorithm that was
analyzed in \cite{ChekuriGKK04} was shown to be scalable for the problem of
minimizing the $\ell_k$-norms of flow time on identical machines for all $ k
>1$. It was suggested in \cite{PruhsST} that determining whether or not $\srpt$
is scalable for the $\ell_k$ norms of flow time on identical machines is another
interesting open question. In this paper we analyze $\srpt$ and show that it is
a scalable algorithm for the $\ell_k$-norm objective on multiple machines. This
shows that not only is $\srpt$ essentially the best possible algorithm for the
objective of average flow time in almost all aspects in the worst case model,
$\srpt$ will also balance the fairness of the schedule when given a small amount
of resource augmentation.

\begin{theorem}\label{thm:lk}
The algorithm $\srpt$ is $(1+\eps)$-speed
$\frac{4}{\eps^2}$-competitive for the $\ell_k$-norms of flow time on
$m$ identical parallel machines for $k \geq 1, 1/2 \geq \eps > 0$.
\end{theorem}

To prove both of these results, we introduce novel potential functions that we
feel capture the structure of \srpt. \srpt\ is a natural algorithm to consider
in many other scheduling models where potential function analysis is commonly
found. We believe that the potential functions introduced here will be useful
for analyzing \srpt\ and similar algorithms in these other settings.\\

\noindent {\bf Related Work:} As mentioned, $\srpt$ is an optimal algorithm for
minimizing average flow time on a single machine. $\srpt$ was the first
algorithm to be analyzed in the worst case model when there are $m$ identical
machines. It was shown by Leonadi and Raz that $\srpt$ is $O(\min\{ \log P, \log
n/m \})$-competitive and there is a matching lower bound on any randomized
algorithm \cite{LeonardiR07}. A simpler analysis of $\srpt$ in the multiple
machine setting can be found in \cite{Leonardi03}. $\srpt$ is
$\left(2-\frac{1}{m} \right)$-speed $1$-competitive and $\srpt$ is the only
algorithm known to be $1$-competitive with any resource augmentation in the
multiple machine setting. Notice that $\srpt$ in the multiple machine setting
could schedule a job on one machine and then later schedule the job on another
machine. That is, $\srpt$ \emph{migrates} jobs between the machines. To
eliminate migration Awerbuch \etal introduced an algorithm that processes each
job on exactly one machine and showed that this algorithm is $O(\min\{\log P,
\log n\})$-competitive \cite{AwerbuchALR02}. A related algorithm was developed
by Chekuri, Khanna, and Zhu that does not migrate jobs and it was shown to be
$O(\min\{\log P, \log n/m\})$-competitive \cite{ChekuriKZ01}. Each of the
previously discussed algorithms hold the jobs in a central pool until they are
scheduled. Avrahami and Azar introduced an algorithm which does not hold jobs in
a central pool, but rather assigns a job to a unique machine as soon as the job
arrives \cite{AvrahamiA07}. They showed that their algorithm is $O(\min\{\log P,
\log n\})$-competitive. Chekuri \etal showed that the algorithm of Avrahami and
Azar is a scalable algorithm \cite{AvrahamiA07, ChekuriGKK04}. For the
$\ell_k$-norms of flow time Chekuri \etal also showed that the algorithm of
Avrahami and Azar is scalable \cite{ChekuriGKK04}.

The analysis in \cite{ChekuriGKK04}, which shows a scalable algorithm for
average flow time on multiple machines, uses a local competitiveness argument.
In a local argument, it is shown that at any time, the increase in the
algorithm's objective function is bounded by a constant factor of the optimal
solution's objective. From the lower bound given above, we know this property
does not hold when $\srpt$ is not given
resource augmentation. With resource augmentation, it is unclear
whether or not this can be shown for $\srpt$ on every input. In this paper, we
avoid a local analysis by using a potential function argument which we discuss
further in the following section.

\section{Preliminaries}
\label{sec:prelim}

Before giving our analysis, we introduce a fair bit of notation. Let $\qs$ and
$\qo$ be the set of jobs alive (released but unsatisfied) at time $t$ in \srpt's
and \opt's schedules, respectively. Let $\ws$ be the set of jobs scheduled for
processing at time $t$ in \srpt's schedule. Let $\ps{i}$ and $\po{i}$ be the
remaining processing times at time $t$ for job $J_i$ in \srpt's and \opt's
schedules, respectively.
Finally, let $\cs{i}$ and $\co{i}$ be the completion time of job $J_i$ in
\srpt's and \opt's schedules, respectively.

Throughout this paper, we will concentrate on bounding~\srpt's
\emph{$k$th power flow time},~$\sum_{i\in[n]}\Paren{\cs{i}-r_i}^k$,
as this is the $\ell_k$-norm of flow time without the outer root.
We will proceed to use $\srpt$ and $\opt$ as functions
of $t$ that return their respective algorithm's accumulated
$k$th power flow time.
In other words, $\srpt(t) = \sum_{i\in[n],t\geq r_i}(\min\Set{\cs{i},t} - r_i)^k$,
and~$\opt(t)$ is defined similarly. When
$\srpt$ or $\opt$ is used as a value without a time specified, it is assumed we
mean their final objective value.

For any job $J_i$ and time $t$, we let $\rs{i}$ be the total volume of work
remaining at time $t$ for every released job with completion time at most
$\cs{i}$ in \srpt's schedule. Precisely,
$$\rs{i}=\sum_{J_j \in \qs, \cs{j} \leq \cs{i}} \ps{j}.$$
We also define
$\vo{i}$ to be the volume of work in \opt's schedule at time $t$ for a subset of
those same jobs, except we only include those
jobs with original processing time at
most $p_i$. Precisely, $$\vo{i} = \sum_{J_j \in \qo, \cs{j} \leq \cs{i},
p_j \leq p_i} \po{j}.$$ We will assume without loss of generality that all
arrival and completion times are distinct by breaking ties arbitrarily but
consistently.

The following lemma will help us to characterize the current status of \srpt\
compared to \opt\ at any point in time. This is a modification of a lemma
given in \cite{MuthukrishnanRSG99, PruhsST}.

\begin{lemma}
\label{lem:status}
At any time $t \geq r_i$, for any sequence of requests $\sigma$,
and for any $ i \in
[n]$, it is the case that $\rs{i} - \vo{i} \leq m p_i$.
\end{lemma}
\begin{proof}
Define $X(i,t)$ to be the sum of the remaining processing times in
\srpt's schedule at time $t$ for jobs with remaining processing time at most
$p_i$ while also contributing to $\rs{i}$. In other words,
$$X(i,t) = \sum_{J_j \in \qs, \cs{j} \leq \cs{i}, \ps{j} \leq p_i} \ps{j}.$$
Every job contributing to $\rs{i}$ must have remaining processing time at most
$p_i$ in order for \srpt\ to schedule it ahead of $J_i$, so we see
$X(i,t) = \rs{i}$ whenever $t \geq r_i$.
Thus is suffices to show that $X(i,t) - \vo{i} \leq m p_i$. If there are $m$ or
fewer jobs contributing to $X(i,t)$ at time $t$ in $\qs$
then the lemma follows easily. Now consider the case where there are more than
$m$ jobs contributing to $X(i,t)$.

Let $t' \geq 0$ be the earliest time before time $t$ such that $\srpt$ always
had at least $m$ jobs contributing to $X(i,t)$ during
$(t',t]$. We will show $X(i,t) - \vo{i} \leq mp_i$. Let $T = \sum_{r_j \in
(t',t], \cs{j} \leq \cs{i}, p_j \leq p_i} p_j $
be the total processing time of jobs that arrive
during $(t',t]$ that are completed by \srpt\ before $J_i$
and have original processing time at most $p_i$. It can be seen that
$X$ will increase by $T$ during $(t',t]$ due to the arrival of jobs. However,
$V^O$ will also increase by $T$ during $(t',t]$ by definition of $V^O$.

The only other change that occurs to $X$ and $V^O$ during $(t',t]$ is due to the
processing of jobs by the algorithm \srpt\ and \opt. Knowing that $\opt$ has $m$
machines of unit speed, $V^O$ can decrease by at most $m(t- t')$ during
$(t,t']$. We also know that during~$(t',t]$, there always exists at least~$m$
jobs with remaining processing time at most~$p_i$ unsatisfied by~\srpt\ that
will be completed by~\srpt\ before job~$J_i$.~\srpt\ always
works on the~$m$ available jobs with earliest completion time,
so this causes $X$ to decrease by at least
$m(t'-t)$ (this even assumes $\srpt$ is not given resource augmentation).
Combining these facts we have the following:

\begin{align*}
  X(i,t)-\vo{i} &\leq \Paren{X(i,t')+T-m(t'-t)} - \Paren{V^O(i,t')+T-m(t'-t)} \\
  &= X(i,t')-V^O(i,t') \\
  &\leq mp_i
\end{align*}

\end{proof}

\subsection{Potential Function Analysis}

For our proofs of the theorems, we will use a potential function argument
\cite{Edmonds00}. In each proof we will define a potential function $\Phi :
[0,\infty) \rightarrow \mathbb {R}$ such that $\Phi(0) = \Phi(\infty) = 0$. We
will proceed to bound discrete and continuous local changes to $\srpt + \Phi$.
These changes may come from the following sources:

\medskip
\noindent \emph{Job Arrival}: Arriving jobs will not affect $\srpt$ but they
will make a change to $\Phi$. The total increase in $\Phi$ over all jobs
arrivals will be bounded by $\delta\opt$ where $\delta$ is a non-negative
constant which may depend on $k$ and $\eps$.\\

\noindent \emph{Job Completion}: Again, job completions will not affect $\srpt$,
but they will make a change to $\Phi$. We will bound these increases by
$\gamma\opt$ where $\gamma$ is a non-negative constant which may depend on
$k$ and $\eps$.\\

\noindent \emph{Running Condition}: This essentially captures everything else.
We will show a bound on the continuous changes in $\srpt + \Phi$ due to the
change in time as well as the changes to each job's remaining processing time.
Surprisingly, we find $\dds + \ddphi
\leq 0$, meaning we can ignore the running condition in our final calculations. 
\medskip

Knowing that $\Phi(\infty)= \Phi(0) = 0$, we have that $\srpt = \srpt(\infty) +
\Phi(\infty)$. This is bounded by the total increase in the arrival and
completion conditions, thus we will have $\srpt \leq (\delta + \gamma) \opt$,
which will complete our analysis.

\section{Total Flow Time}
We consider any job sequence $\sigma$ and assume $\srpt$ is given $(1+\eps)$ speed where $\eps >0 $. We proceed by placing our focus on minimizing the total flow time.   To accomplish this, we will define a potential function with one term for each job being processed such that the following conditions are met:
\begin{itemize}
\item
  Job arrivals and completions do not increase the potential function beyond a
  strong lower bound on \opt.
\item
  Each term has a decreasing component that counteracts the gradual increases in
  \srpt's flow time.
\item
  There may be components of each term that increase, but we can easily bound
  these increases by the decreases from other components.
\end{itemize}

We use the following potential function based on the intuition given above:

$$\Phi(t) = \frac{1}{m \eps} \sum_{J_i \in \qs}
  \Paren{\rs{i} + m\ps{i} - \vo{i}}$$

Now, consider the different changes that occur to \srpt's
accumulated flow time as well as $\Phi$ for any job
sequence $
\sigma$.

\medskip \noindent \emph{Job Arrival:} The event of a job's arrival makes no
change to the accumulated flow time, 
but it can change $\Phi$. Consider the arrival of job $J_i$ at time $t = r_i$.
For any $j \neq i$ such that $J_j \in
\qs$, consider the term
$$\frac{1}{m \eps} \Paren{\rs{j} + m\ps{j} - \vo{j}}$$
in the potential function. The 
arrival of job $J_i$ will change both $\rs{j}$ and $\vo{j}$ equally (either by
$p_i$ or $0$ depending on if $p_j \leq
\ps{i}$) creating no net change in the 
potential function. We do gain a new term in the summation, but this can be
bounded as follows:

\begin{align*}
  &\frac{1}{m\eps} \Paren{\rs{i} + m p_i - \vo{i}} \\
  \leq {} &\frac{1}{m \eps} (2m p_i) \mbox{$\;\;$ By Lemma~\ref{lem:status}} \\
  = {} &\frac{2}{\eps}p_i
\end{align*}

We use the trivial lower bound of~$p_i$ on~$J_i$'s total flow time to see that
the total increase in~$\Phi$ from job arrivals is at most~$\frac{2}{\eps}\opt$.

\medskip \noindent \emph{Job Completion:} Same as above, job completions
make no change to the accumulated 
flow time. Consider the completion of a job $J_i$ by~\opt\ at time~$t = \co{i}$.
For any job~$J_j \in \qs$, the term
$$\frac{1}{m \eps} \Paren{\rs{j} + m\ps{j} - \vo{j}}$$
sees no change as~$J_i$ is already contributing nothing to~$\vo{j}$.

Likewise, consider the completion of job $J_i$ by \srpt\ at 
time $t = \cs{i}$. For any $j \neq i$ such that $J_j \in \qs$, the term
$$\frac{1}{m \eps} \Paren{\rs{j} + m\ps{j} - \vo{j}}$$
sees no change as~$J_i$ is already contributing nothing to~$\rs{j}$.
Unfortunately, we need a more sophisticated argument to bound
in the increase in~$\Phi$ from removing the term
$$\frac{1}{m \eps} \Paren{\rs{i} + m\ps{i} - \vo{i}}.$$

The increase from removing this term is precisely~$\frac{1}{m \eps}\vo{i}$,
because~\srpt\ has completed all jobs contributing to~$\rs{i}$ and~$\ps{i}=0$.
We can use the following scheme to charge this and similar increases to~$\opt$'s
total flow time.
Consider any job~$J_j$ contributing volume to~$\vo{i}$. We know that
if~$r_j < r_i$, we have~$p_j \leq p_i$ by definition of~$V^O$.
Further, if~$r_j \geq r_i$, we have~$p_j \leq p_i^S(r_j)$ by definition
of~$R^S$. In either case,~\srpt\ performs at least~$p_j$ units of work on
job~$J_i$ while~$J_j$ is sitting in~$\opt$'s queue, and this work occurs
over a period of at least~$p_j/(1 + \eps)$ time units. To pay for~$J_j$'s
contribution to~$\frac{1}{m \eps}\vo{i}$, we charge to~$J_j$'s increase
in flow time during this period at a rate of~$\frac{1+\eps}{m \eps}$.

The total charge accrued during this period due to~$J_j$
is at least~$\frac{1+\eps}{m\eps}\frac{p_j}{1+\eps} = \frac{p_j}{m \eps}$.
Summing over
all jobs contributing to~$\vo{i}$, we see that we charge enough. Now we need
to bound our total charge. Observe that any one of these charges to a
job~$J_j$ accrues at~$\frac{1+\eps}{m\eps}$ times
the rate that~$J_j$ is accumulating
flow time. Further,~\srpt\ is working on at most~$m$ jobs at any point in time,
so our combined charges are accruing at~$\frac{1+\eps}{\eps}$ times the rate
that~$J_j$ is accumulating flow time. By summing over all time and jobs, we
conclude that we charge at most~$\frac{1+\eps}{\eps}\opt$, giving us an upper
bound on~$\Phi's$ increase due to~\srpt's job completions.

\medskip \noindent \emph{Running Condition:} We now proceed to show a bound on
$\dds + \ddphi$ at an
arbitrary time $t$ ignoring the arrival and completion of jobs. 
First, note that $$\dds = \sum_{J_i \in \qs} 1.$$

To bound $\ddphi$, we fix some $i$ such that $J_i \in \qs$ and consider $J_i$'s
term in $\Phi$'s summation.

We begin by considering the change due to $\vo{i}$. $\opt$ can only process $m$
jobs at a time, so the $i$th term of $\Phi$ changes at a rate of at most
$$\frac{1}{m\eps} m = \frac{1}{\eps} .$$

Finally, we consider the change due to both $\rs{i}$ and $m \ps{i}$ together and
derive a lower bound on their combined decrease. Neither term can increase, so
we accomplish this by finding a lower bound on the decrease of one or the other.
Suppose~\srpt\ is processing job $J_i$ (using
$(1 + \eps)$ 
speed) at time $t$. If this is the case, $m \ps{i}$ decreases at a rate
of $m(1 + \eps)$. If job $J_i$ is not being processed, then 
there are $m$ other jobs in $\ws$ being processed instead. 
By definition, these jobs are contributing their volume to~$\rs{i}$,
and we see it decreases at a rate of $m(1 + \eps)$. 
Considering both terms together, we find an
upper bound for their contribution to
$\Phi$'s rate of change which is
$$\frac{1}{m \eps} \Paren{-m(1+\eps)} = -\frac{1}{\eps} - 1.$$

By summing over the above rates of change, we see everything
cancels out to $0$.
Summing over all jobs gives us $\dds(t) +
\ddphi(t) \leq 0$.
Integrating the left hand side over all time, we
see $\srpt$ and $\Phi$ together \emph{do not increase}
if we only consider events other than the arrival and completion of jobs.

\medskip \noindent \textbf{Final Analysis:} Using the framework described in
Section \ref{sec:prelim} and the above analysis, we
see~$\srpt~\leq~\frac{4}{\eps}\opt$. This concludes the proof of
Theorem~\ref{thm:avg}. \qed


\section{$\ell_k$-Norms of Flow Time}
In this section we focus on minimizing the $\ell_k$-norms of flow
time.  Consider any job sequence $\sigma$ and assume that $\srpt$ is given $(1+\eps)$-speed where~$1/2\geq\eps>0$. 
We use a somewhat different potential function that includes extra components
meant to reflect the increasing speed at which alive jobs contribute to $k$th
power flow time. We use the following potential function to directly
bound~\srpt's~$k$th power flow time:

$$\Phi(t) = \frac{1}{(1-\eps)^k} \sum_{J_i \in \qs}
  \Paren{\max \Set{t - r_i + \frac{1}{m\eps}\Paren{\rs{i}+m\ps{i}-\vo{i}},0}}^k
  - \sum_{J_i \in \qs} (t - r_i)^k$$

Consider any job sequence $\sigma$.

\medskip \noindent \emph{Job Arrival:} Consider the arrival of job $J_i$ at time
$t = r_i$. Again, no change occurs
to the objective function. Also, as in the case for standard flow time, no
change will occur to the $J_j$th term of the
potential function for any $j \neq i$. 
However, a new term is added to the 
summation in the potential function. The increase in $\Phi$ due to this new term
is at most

\begin{align*}
  &\frac{1}{(1-\eps)^k}
    \Paren{\frac{1}{m\eps}\Paren{\rs{i} + m p_i - \vo{i}}}^k \\
  \leq {} &\frac{1}{(1 - \eps)^k} \Paren{\frac{1}{m\eps}(2m p_i)}^k
    \mbox{$\;\;$ By Lemma~\ref{lem:status}} \\
  \leq {} &\Paren{\frac{2}{\eps(1-\eps)}}^k (p_i)^k.
\end{align*}

The value~$(p_i)^k$ is a trivial lower bound on~$J_i$'s~$k$th power flow time,
so we can bound the total increase in~$\Phi$ due to job arrivals
by~$\Paren{\frac{2}{\eps(1-\eps)}}^k \opt$.

\medskip \noindent \emph{Job Completion:} Again, the only effect of job
completion we are concerned with is the increase of each job~$J_i$'s term
in~$\Phi$ when~\srpt\
completes~$J_i$ at time~$t=\cs{i}$. The increase from this occurrence is
$$(t-r_i)^k - \frac{1}{(1-\eps)^k}
  \Paren{\max\Set{t-r_i+\frac{1}{m\eps}\vo{i},0}}^k$$
We will use the following lemmas.

\begin{lemma}
  For any job~$J_i \in \qs$, if~$\vo{i} \leq m\eps^2(t-r_i)$ then
  $$(t-r_i)^k - \frac{1}{(1-\eps)^k}
    \Paren{\max\Set{t-r_i+\frac{1}{m\eps}\vo{i},0}}^k \leq 0.$$
\end{lemma}

\begin{proof}
  Note that hypothesis cannot apply when
  $t-r_i+\frac{1}{m\eps}\Paren{\rs{i}+m\ps{i}-\vo{i}} < 0$. This is because
  $$t-r_i+\frac{1}{m\eps}\Paren{\rs{i}+m\ps{i}-\vo{i}} \geq (1-\eps)(t-r_i)$$
  which is non-negative for all $t \geq r_i, \eps \leq 1$.
  Given the assumption that~$\vo{i} \leq m\eps^2(t-r_i)$, we have
  \begin{align*}
    &(t-r_i)^k-\frac{1}{(1-\eps)^k}\Paren{t-r_i+\frac{1}{m\eps}\vo{i}}^k \\
    \leq{}&(t-r_i)^k-\frac{1}{(1-\eps)^k}\Paren{(1-\eps)(t-r_i)}^k \\
    ={}&0.
  \end{align*}
\end{proof}

\begin{lemma}
  For any job~$J_i \in \qs$, if~$\vo{i} > m\eps^2(t-r_i)$ then
  $$(t-r_i)^k - \frac{1}{(1-\eps)^k}
    \Paren{\max\Set{t-r_i+\frac{1}{m\eps}\vo{i},0}}^k
    \leq \Paren{\frac{1}{\eps^2}}^k\Paren{\frac{1}{m}\vo{i}}^k.$$
\end{lemma}

\begin{proof}
  We will ignore the negative term from the expression. Given the assumption
  that $\vo{i}~>~m\eps^2\Paren{t-r_i}$, we have
  \begin{align*}
    (t-r_i)^k &\leq \Paren{\frac{1}{m\eps^2}\vo{i}}^k \\
    &= \Paren{\frac{1}{\eps^2}}^k\Paren{\frac{1}{m}\vo{i}}^k.
  \end{align*}
\end{proof}

Based on these lemmas, we see the total increase to~$\Phi$ from job completions
is bounded
by
$$\sum_{i\in[n]}\Paren{\frac{1}{\eps^2}}^k\Paren{\frac{1}{m}\voc{i}}^k.$$
The following lemma, which we will prove later, implies that this bound is at
most~$\Paren{\frac{1+\eps}{\eps^2}}^k\opt$.

\begin{lemma}
\label{lem:l_k-charge}
  We have
  $$\sum_{i\in[n]}\Paren{\frac{1}{m}\voc{i}}^k \leq (1+\eps)^k\opt.$$
\end{lemma}

\medskip \noindent \emph{Running Condition:} We now ignore the arrival and
completion of jobs and consider the change in the $k$th power flow time as well
as $\Phi$ due to other events. Consider any time $t$. Note that
$$\dds(t) = \sum_{J_i \in \qs} k \cdot (t - r_i)^{k-1}. $$

Now, fix some $i$ such that $J_i \in \qs$. We will examine the
contribution of the
$J_i$th term to $\dds$. We will begin by assuming $t - r_i + \frac{1}{m \eps}
\Paren{\rs{i} + m \ps{i} - \vo{i}} > 0$ and consider the other case later.

First, consider how the change in $t$ affects this term while keeping the
dependent variables fixed. The rate of change is at most
$$ \frac{k}{(1-\eps)^k}
  \Paren{t - r_i + \frac{1}{m\eps}\Paren{\rs{i}+m\ps{i}-\vo{i}}}^{k-1}
  - k(t-r_i)^{k-1}.$$

Next we consider the change due to $\vo{i}$. In the worst case, \opt\ works on
$m$ jobs at time $t$ so the rate of increase in $\Phi$ due to the change in
$\vo{i}$ is at most

\begin{align*}
  & \frac{km}{m\eps}\frac{1}{(1-\eps)^k}\Paren{t-r_i+\frac{1}{m\eps}
    \Paren{\rs{i}+m\ps{i}-\vo{i}}}^{k-1} \\
  = {} & \frac{k}{\eps(1-\eps)^k}\Paren{t-r_i+\frac{1}{m\eps}
    \Paren{\rs{i}+m\ps{i}-\vo{i}}}^{k-1}
\end{align*}

Now consider the change in $\Phi$ due to $\rs{i} + m\ps{i}$. As in the
average flow time argument, this sum decreases at a rate of at least $(1 +
\eps)m$, so these terms cause $\Phi$ to change at a rate of at most

\begin{align*}
  &-\frac{k(1+\eps)m}{m\eps}\frac{1}{(1-\eps)^k}
    \Paren{t-r_i+\frac{1}{m\eps}\Paren{\rs{i}+m\ps{i}-\vo{i}}}^{k-1} \\
  ={}&-\frac{k}{\eps(1-\eps)^k}
    \Paren{t-r_i+\frac{1}{m\eps}\Paren{\rs{i}+m\ps{i}-\vo{i}}}^{k-1} \\
  &\quad-\frac{k}{(1-\eps)^k}
    \Paren{t-r_i+\frac{1}{m\eps}\Paren{\rs{i}+m\ps{i}-\vo{i}}}^{k-1}.
\end{align*}

Summing over the above terms shows that $J_i$ contributes at most~$0$ to
$\srpt+\Phi$'s rate of change.

We have yet to consider the case when $t - r_i + \frac{1}{m \eps} \Paren{\rs{i}
+ m \ps{i} - \vo{i}} \leq 0$. The above arguments concerning
the running condition and job arrivals show this term to be
non-increasing. Further, we see that once the expression
$$\max \Set{t-r_i+\frac{1}{m\eps}\Paren{\rs{i}+m\ps{i}-\vo{i}},0}$$
hits~$0$, it will never leave that value.

We now consider the various sources of change in $\Phi$'s $i$th term when
the above expression equals~$0$ by simply plugging $0$ into
the above inequalities. Changes in $t$ contribute at
most~$-k(t-r_i)^{k-1}$. Also, changes in $\rs{i}$, $m
\ps{i}$, and $\vo{i}$ have no effect. Summing, we still get~$0$ as an upper
bound on~$J_i$'s contribution to~$\srpt+\Phi$'s rate of change. Summing over
all jobs and integrating over time,
we see this bound holds for the running condition's
overall contribution to~$\srpt+\Phi$.

\medskip \noindent {\bf Final Analysis:} Using
the framework discussed in
Section~\ref{sec:prelim}  
and the arrival, completion, and running conditions shown in this section, we
have that
$$\srpt \leq \Paren{\Paren{\frac{2}{\eps(1-\eps)}}^k
  + \Paren{\frac{1+\eps}{\eps^2}}^k}\opt.$$
By taking the outer $k$th
root of the $\ell_k$-norm flow time and assuming $\eps < 1/2$, we
derive Theorem~\ref{thm:lk}. \qed

\section{Proof of Lemma~\ref{lem:l_k-charge}}

In this section, we prove Lemma~\ref{lem:l_k-charge}. Namely, if~\srpt\ is
running~$m$ machines of speed~$(1+\eps)$ while~\opt\ is running~$m$ machines
of unit speed, we have
$$\sum_{i\in[n]}\Paren{\frac{1}{m}\voc{i}}^k \leq (1+\eps)^k \opt$$
for the metric of $k$th power flow time. We will use a charging scheme to prove
the lemma.

Fix some job~$J_i$ and let~$\sjobs{i}$ denote the set of jobs that contribute
to~$\voc{i}$. We charge the following to each~$J_j \in \sjobs{i}$:

$$\Paren{\frac{1}{m}
  \Paren{\voc{i}-\sum_{J_a\in\sjobs{i},r_a<r_j}p_a^O(\cs{i})}}^k
  -\Paren{\frac{1}{m}
  \Paren{\voc{i}-p_j^O(\cs{i})-\sum_{J_a\in\sjobs{i},r_a<r_j}p_a^O(\cs{i})}}^k
$$

By considering the jobs in~$\sjobs{i}$ in order of increasing arrival time,
we see the charges form a telescoping sum that evaluates to

$$\Paren{\frac{1}{m}\voc{i}}^k
  - \Paren{\frac{1}{m}\Paren{\voc{i}-\sum_{J_j\in\sjobs{i}}p_j^O(\cs{i})}}^k
  = \Paren{\frac{1}{m}\voc{i}}^k.$$

Now our goal is to show that we charge at most~$(1+\eps)^k(\co{j}-r_j)^k$
in total to any job~$J_j$. Let $\tjobs{j}~=~\Set{J_i \mid J_j\in\sjobs{i}}$,
the set of jobs whose completion causes us to charge some amount to~$J_j$.
Consider the charge on~$J_j$ due to the completion of~$J_i\in\tjobs{j}$.

\begin{lemma}
\label{lem:work-leq}
  We have
  $$\frac{1}{(1+\eps)m}
    \Paren{\voc{i}-\sum_{J_a\in\sjobs{i},r_a<r_j}p_a^O(\cs{i})}
    \leq \co{j}-r_j-\frac{1}{(1+\eps)m}
    \sum_{J_a\in\tjobs{j},\cs{a}>\cs{i}}p_j^O(\co{a}).$$
\end{lemma}

\begin{proof}
  We will account for work done by~\srpt\ during~$[r_j,\co{j}]$ in two stages
  and use the result to derive the inequality. First, consider any
  job~$J_a \in \tjobs{j}$ with~$\cs{a}>\cs{i}$.
  We know~\srpt\ gave higher priority to~$J_j$
  than~$J_a$, because~$J_j$ is included in~$\voc{a}$. As seen in the completion
  condition arguments for total flow time, we know~\srpt\ did~$p_j$ volume of
  work on job~$J_a$ during~$[r_j,\cs{a}]$.
  Namely, we have~$p_j\leq p_a$ when~$r_j < r_a$
  and~$p_j\leq p_a^S(r_j)$ when~$r_j \geq r_a$ by definition of~$\voc{i}$.
  Therefore, we have at least
  $\sum_{J_a\in\tjobs{j},\cs{a}>\cs{i}}p_j
  \geq \sum_{J_a\in\tjobs,\cs{a}>\cs{i}}p_j^O(\co{a})$ volume of work
  done by~\srpt\ during~$[r_j,\co{j}]$.

  Next, we note that an additional
  $\voc{i}-\sum_{J_a\in\sjobs{i},r_a<r_j}p_a^O(\cs{i})$
  volume of work must be completed by~\srpt\ during~$[r_j,\co{j}]$. This is
  because~\srpt\ completed the jobs being counted in the above expression by
  time~$\cs{i}\leq\co{j}$ and these jobs arrived after time~$r_j$. Further,
  we are not counting the work in the above paragraph a second time,
  because no job~$J_a$ with~$\cs{a}>\cs{i}$ can count toward~$\rsc{i}$
  or~$\voc{i}$ by definition of $R^S$ and $V^O$.

  We know~\srpt\ has~$m$ machines of speed~$1+\eps$, so the soonest~\srpt\ can
  complete the above mentioned work is
  $$r_j+\frac{1}{(1+\eps)m}\Paren{\sum_{J_a\in\tjobs,\cs{a}>\cs{i}}p_j^O(\co{a})
  +\voc{i}-\sum_{J_a\in\sjobs{i},r_a<r_j}p_a^O(\cs{i})}.$$
  This expression is at most~$\co{j}$. The lemma follows by simple algebra.
\end{proof}

Now we are ready to prove a bound on the amount charged to~$J_j$. The total
amount charged is
$$\sum_{J_i\in\tjobs{j}}\Brack{
  \Paren{\frac{1}{m}
  \Paren{\voc{i}-\sum_{J_a\in\sjobs{i},r_a<r_j}p_a^O(\cs{i})}}^k
  -\Paren{\frac{1}{m}
  \Paren{\voc{i}-p_j^O(\cs{i})-\sum_{J_a\in\sjobs{i},r_a<r_j}p_a^O(\cs{i})}}^k
  }.$$
Using Lemma~\ref{lem:work-leq} and the convexity of~$x^k$ 
for~$k\geq1$ (where $x$ is any positive number), we can upper bound
this by
\begin{align*}
  &\sum_{J_i\in\tjobs{j}}\bigg[
    \Paren{(1+\eps)(\co{j}-r_j)-
    \frac{1}{m}\sum_{J_a\in\tjobs{j},\cs{a}>\cs{i}}p_j^O(\co{a})}^k \\
  &\quad\quad-\Paren{(1+\eps)(\co{j}-r_j)-\frac{1}{m}\Paren{-p_j^O(\cs{i})
    -\sum_{J_a\in\tjobs{j},\cs{a}>\cs{i}}p_j^O(\co{a})}}^k\bigg].
\end{align*}

Again, it can be seen that this is a telescoping sum by considering terms
in order of decreasing completion time. By the arguments given in the proof
of Lemma~\ref{lem:work-leq}, we
see
$$\frac{1}{(1+\eps)m}\sum_{J_a\in\tjobs{j}}p_j^O(\co{a})\leq\co{j}-r_j,$$
giving us a lower bound of~$0$ for the last negative term in the
telescoping sum. Therefore, the total charged to~$J_j$ is at
most~$\Paren{(1+\eps)(\co{j}-r_j)}^k$. Summing over all jobs, we see the total
amount charged is at most
$$\sum_{j\in[n]}(1+\eps)^k(\co{j}-r_j)^k=(1+\eps)^k\opt,$$
which implies the lemma. \qed

\section{Conclusion}

We have shown \srpt\ to be $(1 + \eps)$-speed $O(1)$-competitive for both
average flow time and further for the $\ell_k$-norms of flow time on $m$
identical machines. This combined with previous work shows that \srpt\ is the
best possible algorithm in many aspects for scheduling on $m$ identical
machines. It is known that $\srpt$ is $(2-\frac{1}{m})$-speed 1-competitive on
multiple machines . Further, it is known that no $(\frac{22}{21} - \eps)$-speed
online algorithm is $1$-competitive \cite{PhillipsSTW02}. It remains an
interesting open question to determine the minimum speed needed for an algorithm
for be $1$-competitive on $m$ identical machines. \\

\noindent {\bf Acknowledgements:} We would like to thank the anonymous
reviewers for their suggestions on improving this paper.

\bibliographystyle{alpha}
\bibliography{SRPT}

\end{document}